\documentclass[12pt]{article}

\usepackage{amsmath,amssymb,amsfonts}
\usepackage{algorithmic}
\usepackage{graphicx}
\usepackage{textcomp}
\usepackage{algorithm}
\usepackage{subcaption}
\usepackage{bm}
\usepackage{amsthm}
\usepackage{float}
\usepackage{placeins}
\usepackage{geometry}
\usepackage{hyperref}
\usepackage{adjustbox}

\newtheorem{theorem}{Theorem}
\newtheorem{lemma}{Lemma}
\newtheorem{corollary}{Corollary}

\title{Closed-Form and Constant-Time New-Source Selection for Fault-Tolerant Broadcasting in Dense Gaussian Networks}

\author{Bader Albader\\
\small Department of Computer Science, Kuwait University, Kuwait\\
\small \texttt{albader@cs.ku.edu.kw}}

\date{}

\begin{document}

\maketitle

\begin{abstract}
Fault-tolerant broadcasting in dense Gaussian networks can be recovered by re-rooting the broadcast at a new source that is at maximum graph distance from the faulty nodes. This paper extends the published re-rooting framework by replacing its boundary-search source-selection step with a quotient-lattice-aware algebraic construction. The first contribution is a constant-time counting method for valid new sources. The counting problem is formulated as an intersection of two diameter-$k$ boundary sets in the Gaussian quotient. A compact piecewise expression is retained for the local, unshifted boundary intersection, and the exact quotient-network count is obtained by a fixed union of side-pair intervals over the nine relevant quotient-lattice copies. This gives a closed-form constant-size counting procedure without scanning either the network or the boundary. The second contribution is a shifted direct selector for two arbitrary faulty nodes. Given faulty nodes $A$ and $B$, the problem is translated to $C=\operatorname{mod}_{G_k}(B-A)$, and the selector finds a point $P$ satisfying $d(P,0)=d(P,C)=k$. For each of nine quotient-lattice shifts, the selector checks sixteen signed linear systems. Nonparallel systems are solved using Cramer's rule, while parallel systems are handled by interval-endpoint selection. Thus, at most $9\times16=144$ shifted sign cases are evaluated, giving $O(1)$ source selection under the standard word-RAM model. Computational validation reports zero count mismatches over 26,623 tested nodes, 500,000 valid outputs over 500,000 sampled fault pairs, and 40,000 successful re-rooted broadcast trials. Runtime results show that the shifted selector has a small fixed overhead for small networks but remains nearly stable as $k$ increases, achieving a $5.92\times$ speedup over boundary search at $k=200$. These results strengthen the re-rooting approach by making new-source selection algebraic, bounded, and independent of the network size.
\end{abstract}


\section{Introduction}

Interconnection networks play an important role in the design of parallel and distributed computing systems. Their topology directly affects communication latency, fault tolerance, routing complexity, and broadcasting efficiency~\cite{grama2003parallel,duato2002interconnection,dally2004principles}. Among many well-known network structures, circulant-based networks and Gaussian networks have received attention because they combine regularity, symmetry, and relatively small diameter~\cite{flahive2010topology,martinez2006dense,beivide2005gaussian}. These properties make them useful models for studying communication algorithms, routing strategies, and fault-tolerant information dissemination.

Dense Gaussian networks form a special class of degree-four interconnection networks constructed over Gaussian integers~\cite{martinez2006dense,beivide2005gaussian}. A dense Gaussian network with parameter $k$, denoted by $G_k$, contains
\[
N = k^2 + (k+1)^2
\]
nodes and has diameter $k$. Each node can be represented either by a Gaussian integer coordinate or by an equivalent integer label modulo $N$. This algebraic structure provides a compact way to describe routing paths, coordinate differences, and wrap-around connections. Related algebraic approaches using Gaussian integers and circulant structures have also been studied in connection with domination, routing, and scalable network construction~\cite{martinez2005perfect,vallejo2008hierarchical}. As a result, dense Gaussian networks offer a useful setting for the study of broadcasting and fault-tolerant communication.

Broadcasting is a fundamental communication operation in which one source node sends information to all other nodes in the network. In a fault-free dense Gaussian network, the regular structure of $G_k$ allows a source node to reach all other nodes within $k$ steps. However, when one or more nodes fail, some broadcast paths may be interrupted. This can prevent the original broadcast tree from covering all non-faulty nodes. Fault-tolerant communication is therefore a major concern in interconnection networks and network-on-chip architectures~\cite{kliazovich2010survey,pasricha2008onchip,flich2010designing}. A fault-tolerant broadcasting method must not only detect the effect of faulty nodes, but also provide an efficient mechanism for restoring coverage.

The published re-rooting-based broadcasting method addresses this problem by selecting a new source node when faults occur~\cite{albader2026rerooting}. The main idea is to restart or continue the broadcast from a carefully selected node that is farthest from the faulty nodes in terms of graph distance. In particular, for two faulty nodes $A$ and $B$, a valid new source $NS$ should satisfy
\[
d(NS,A)=k
\]
and
\[
d(NS,B)=k.
\]
Such a node allows the recovered broadcast to begin from a position that avoids the faulty region and preserves the maximum-distance structure of the network.

Although the re-rooting approach provides an effective recovery mechanism, the selection of a valid new source raises an important mathematical and algorithmic question. Existing treatment mainly establishes the existence of a suitable new source or obtains one through search-based procedures. A direct closed-form method for counting and selecting valid new sources is still needed. Without such a method, the new-source selection step may depend on scanning network nodes or testing boundary candidates, which obscures the algebraic structure of the problem and may increase the cost of recovery.

The present work is a direct extension of the re-rooting-based method introduced in~\cite{albader2026rerooting}. In that work, the re-rooting framework was proposed, the existence of a common distance-$k$ node was proved for any two faulty nodes, and an $O(k)$ boundary-search procedure was used to find a valid new source. In contrast, the present paper does not propose a new broadcasting model. Instead, it strengthens the mathematical and algorithmic foundation of the source-selection step by deriving the exact number of valid new sources and by replacing the boundary-search procedure with a constant-time algebraic construction.

This paper extends the re-rooting framework by developing a quotient-lattice-aware counting and selection theory for valid new sources in dense Gaussian networks. The first part of the paper studies the number of valid new sources. For a given node
\[
A=(a_x,a_y),
\]
we count the number of nodes $p$ satisfying
\[
d(0,p)=k
\]
and
\[
d(A,p)=k.
\]
These nodes are exactly the candidates that are simultaneously at maximum graph distance from the origin and from $A$. By expressing the position of $A$ using
\[
a=\max\{|a_x|,|a_y|\}
\]
and
\[
b=\min\{|a_x|,|a_y|\},
\]
we derive a compact piecewise formula for the local boundary-intersection count and then extend it to an exact quotient-lattice-aware count using a fixed set of side-pair interval computations. This shows that every nonzero node has at least four valid new sources and also accounts for wrap-around boundary intersections.

The second part of the paper uses this structure to develop a direct method for two arbitrary faulty nodes. Given faulty nodes $A$ and $B$, the problem is translated to an equivalent problem involving the origin and the difference node
\[
C=\operatorname{mod}_{G_k}(B-A).
\]
A valid point $P$ is then found such that
\[
d(P,0)=k
\]
and
\[
d(P,C)=k.
\]
The desired new source is obtained by translating back:
\[
NS=\operatorname{mod}_{G_k}(A+P).
\]
The proposed method solves this problem by considering a fixed set of quotient-lattice copies and a fixed set of sign configurations. Specifically, it checks nine lattice shifts and sixteen sign tuples per shift, for at most $9\times16=144$ algebraic cases. Since this bound is independent of $k$ and $N$, the selection procedure runs in $O(1)$ time under the standard word-RAM model.

The main contributions of this paper are summarized as follows:
\begin{itemize}
    \item We formulate the valid new-source problem as an intersection problem between two $k$-step boundary sets in $G_k$.
    \item We derive a compact piecewise formula for the local boundary-intersection count and an exact quotient-lattice-aware constant-time counting formula for the full Gaussian quotient network.
    \item We prove that every nonzero node has at least four valid new sources, strengthening the existence argument used in re-rooting-based fault-tolerant broadcasting.
    \item We develop a quotient-lattice-aware direct algebraic construction for selecting a valid new source for two arbitrary faulty nodes. The construction handles wrap-around boundary copies, nonparallel signed systems, and parallel signed systems without scanning the network or the boundary.
    \item We show that the proposed shifted direct-selection method runs in $O(1)$ time because it checks at most $9\times16=144$ algebraic cases and performs only constant-time arithmetic in each case, independent of the network size.
\end{itemize}

The rest of this paper is organized as follows. Section~\ref{sec:related-work} reviews related work on Gaussian networks, circulant topologies, fault-tolerant communication, and algebraic routing methods. Section~\ref{sec:preliminaries} introduces the required notation and background on dense Gaussian networks. Section~\ref{sec:counting} derives the local piecewise count and the exact quotient-lattice-aware counting formula for valid new sources. Section~\ref{sec:direct-selection} presents the shifted constant-time construction for two arbitrary faulty nodes. Section~\ref{sec:complexity} compares the proposed method with search-based alternatives. Section~\ref{sec:validation} presents computational validation of the counting formula and the direct-selection algorithm. Section~\ref{sec:conclusion} concludes the paper and discusses future work. Finally, Appendix~\ref{app:side-pair-counting} gives the side-pair interval formulas used by the quotient-lattice-aware counter.

\section{Related Work}
\label{sec:related-work}

Gaussian and circulant-based interconnection networks have been widely studied because they provide regular structures, small diameter, and algebraic representations that are useful for routing and broadcasting. Flahive and Bose studied the topology of Gaussian and Eisenstein--Jacobi interconnection networks and showed how quotient rings of Gaussian and Eisenstein--Jacobi integers can be used to construct interconnection networks with desirable topological properties~\cite{flahive2010topology}. Dense Gaussian networks were later studied as degree-four circulant-based topologies suitable for on-chip multiprocessors, where their smaller diameter and two-dimensional labeling provide advantages over traditional torus-based structures~\cite{martinez2006dense}. Related work on Gaussian interconnection networks also introduced coordinate-based representations that simplify analysis, routing, and broadcasting in these networks~\cite{beivide2005gaussian}. In addition, the connection between circulant graphs, Gaussian integers, and domination problems further demonstrates the usefulness of algebraic methods in this family of networks~\cite{martinez2005perfect}.

Hierarchical and scalable versions of these networks have also been considered. Vallejo, Mart\'inez, and Beivide studied hierarchical topologies for large-scale two-level networks, extending the use of Gaussian and circulant-based structures to larger systems~\cite{vallejo2008hierarchical}. More generally, interconnection networks are central to parallel and distributed computing because network topology directly affects routing cost, communication latency, bisection bandwidth, and fault tolerance~\cite{grama2003parallel}. Standard treatments of interconnection networks emphasize that regularity, symmetry, degree, and diameter are important design factors for scalable communication systems~\cite{duato2002interconnection,dally2004principles}.

Fault tolerance and reliable communication are major concerns in interconnection networks and network-on-chip architectures. Kliazovich, Granelli, and Miorandi surveyed fault-tolerant communication in network-on-chip architectures and emphasized the importance of maintaining communication under link or node failures~\cite{kliazovich2010survey}. Pasricha and Dutt discussed on-chip communication architectures and showed how topology and routing decisions affect system-level performance and reliability~\cite{pasricha2008onchip}. Flich and Bertozzi also studied network-on-chip design in the nanoscale era, where reliability, scalability, and routing efficiency are critical concerns~\cite{flich2010designing}.

Broadcasting and spanning-tree construction are also closely related to the present work. Wu and Huang developed a distributed algorithm for constructing independent spanning trees in parallel systems, which is relevant because independent or alternative spanning structures can improve communication reliability~\cite{wu1999distributed}. Saad and Schultz studied topological properties of hypercubes, including structural properties important for communication algorithms in parallel systems~\cite{saad1988topological}.

Recent work has continued to study circulant and Gaussian-related topologies for networks-on-chip. Romanov et al. developed routing algorithms for networks-on-chip based on two-dimensional optimal circulant topologies, showing the continued relevance of circulant structures for efficient routing in NoC environments~\cite{romanov2020routing}. Monakhova, Monakhov, and Romanov proposed routing algorithms for optimal degree-four circulant networks based on relative addressing, which is particularly relevant because their approach uses algebraic and coordinate-based routing ideas in degree-four circulant networks~\cite{monakhova2022relative}. Song et al. studied Gaussian-based optical networks-on-chip and reported performance advantages related to diameter and hop distance, further supporting the use of Gaussian-based topologies in modern communication architectures~\cite{song2020gaussian}.

The present paper differs from these prior works in its objective. Existing work mainly studies topology, routing, broadcasting, or general fault-tolerant communication. Recent re-rooting-based broadcasting work showed that a broadcast in a dense Gaussian network can be recovered by selecting a valid new source when faults occur~\cite{albader2026rerooting}. However, the new-source selection step remained primarily an existence or search-based step. This paper addresses that gap by deriving a closed-form formula for the number of valid new sources and by developing a constant-time algebraic method for selecting one.

\section{Preliminaries and Network Model}
\label{sec:preliminaries}

This section introduces the notation and basic properties used throughout the paper. The network considered in this work is the dense Gaussian network $G_k$, where $k$ is a positive integer parameter. The node set of $G_k$ can be represented by Gaussian integer coordinates with wrap-around equivalence. The total number of nodes is
\[
N = k^2 + (k+1)^2.
\]
The graph has diameter $k$, meaning that every node can be reached from any other node in at most $k$ steps.

\subsection{Node Representation}

A node in $G_k$ is represented by an ordered pair
\[
A=(a_x,a_y),
\]
which corresponds to the Gaussian integer
\[
a_x + a_y i.
\]
Throughout the paper, we use coordinate notation and Gaussian-integer notation interchangeably whenever the meaning is clear.

Because $G_k$ is a wrap-around network, different coordinate pairs may represent the same node. Therefore, whenever a coordinate expression produces a point outside the chosen representative region of $G_k$, it is reduced modulo the Gaussian network. We denote this reduction by
\[
\operatorname{mod}_{G_k}(\cdot).
\]

\subsection{Graph Distance}

For two nodes \(U,V\in G_k\), the graph distance between them is denoted by
\[
d(U,V).
\]
This distance is the minimum number of graph edges required to move from \(U\) to \(V\) in \(G_k\). Since \(G_k\) has diameter \(k\), we have
\[
d(U,V)\leq k
\]
for all nodes \(U,V\in G_k\).

Because \(G_k\) is a wrap-around network, a coordinate difference may have several equivalent Gaussian-integer representatives. Throughout this paper, whenever a difference node is written as
\[
X=\operatorname{mod}_{G_k}(U-V)=(x,y),
\]
we assume that $\operatorname{mod}_{G_k}(\cdot)$ returns the canonical reduced representative whose Manhattan length is minimum among all equivalent representatives. Therefore, the graph distance from the origin is computed by
\[
d(0,X)=|x|+|y|.
\]
Equivalently,
\[
d(U,V)=\left|x\right|+\left|y\right|,
\quad
\text{where }
(x,y)=\operatorname{mod}_{G_k}(U-V).
\]

Thus, a reduced node \(X=(x,y)\) lies on the \(k\)-step boundary of the origin precisely when
\[
|x|+|y|=k.
\]

The distance function is translation invariant. That is, for any three nodes \(U,V,T\in G_k\),
\[
d(U,V)
=
d\!\left(
\operatorname{mod}_{G_k}(U+T),
\operatorname{mod}_{G_k}(V+T)
\right).
\]
Equivalently,
\[
d(U,V)
=
d\!\left(
\operatorname{mod}_{G_k}(U-V),
0
\right).
\]
This property is essential for reducing the two-fault new-source problem to the case where one faulty node is the origin.

\subsection{Boundary Nodes}

The boundary of radius $k$ around the origin is defined as
\[
\partial G_k(0)=\{P\in G_k : d(0,P)=k\}.
\]
More generally, the boundary of radius $k$ around a node $A$ is
\[
\partial G_k(A)=\{P\in G_k : d(A,P)=k\}.
\]

Since the diameter of $G_k$ is $k$, boundary nodes are the nodes that are farthest from the center node. In the context of re-rooting-based broadcasting, such farthest nodes are important because they preserve the maximum-distance structure needed for full broadcast recovery.

For the origin, the number of boundary nodes is
\[
|\partial G_k(0)|=4k.
\]
These nodes form the set of possible farthest nodes from the origin.

\subsection{Valid New Sources}

Let $A\in G_k$ be a node. A node $P\in G_k$ is called a valid new source with respect to $0$ and $A$ if
\[
d(0,P)=k
\]
and
\[
d(A,P)=k.
\]
Thus, $P$ must lie in the intersection of two boundary sets:
\[
P\in \partial G_k(0)\cap \partial G_k(A).
\]
The number of such nodes is denoted by
\[
\operatorname{count}(A,k)
=
\left|
\partial G_k(0)\cap \partial G_k(A)
\right|.
\]

For two arbitrary faulty nodes $A,B\in G_k$, a node $NS\in G_k$ is called a valid new source if
\[
d(NS,A)=k
\]
and
\[
d(NS,B)=k.
\]
The goal of the direct-selection problem is to find such an $NS$ without scanning all nodes of $G_k$.

\subsection{Difference Node for Two Faults}

Given two faulty nodes $A,B\in G_k$, define the difference node
\[
C=\operatorname{mod}_{G_k}(B-A).
\]
If
\[
C=(c,d),
\]
then the problem of finding a node $NS$ satisfying
\[
d(NS,A)=k
\]
and
\[
d(NS,B)=k
\]
can be translated into the equivalent problem of finding a node $P$ satisfying
\[
d(P,0)=k
\]
and
\[
d(P,C)=k.
\]
Once such a point $P$ is found, the required new source is obtained by shifting back:
\[
NS=\operatorname{mod}_{G_k}(A+P).
\]

This translation is valid because of distance translation invariance. Indeed,
\[
d(NS,A)
=
d\bigl(\operatorname{mod}_{G_k}(A+P),A\bigr)
=
d(P,0),
\]
and since $B=A+C$ in $G_k$,
\[
d(NS,B)
=
d\bigl(\operatorname{mod}_{G_k}(A+P),\operatorname{mod}_{G_k}(A+C)\bigr)
=
d(P,C).
\]
Therefore, if $P$ is at distance $k$ from both $0$ and $C$, then $NS$ is at distance $k$ from both faulty nodes $A$ and $B$.

\subsection{Integer Labeling}

The coordinate representation can also be connected to an integer labeling modulo $N$. Let
\[
\phi:G_k\rightarrow \mathbb{Z}_N
\]
be defined by
\[
\phi(x+yi)\equiv kx+(k+1)y \pmod{N},
\]
where
\[
N=k^2+(k+1)^2.
\]
This mapping assigns an integer label modulo $N$ to each node of $G_k$. Moreover, coordinate differences are preserved under the labeling in the sense that
\[
\phi(U-V)\equiv \phi(U)-\phi(V)\pmod{N}.
\]
This algebraic representation is useful when implementing the proposed direct-selection method, since the difference node $C=B-A$ can be computed either in coordinate form or through the corresponding modular labels.

\section{Counting the Number of Valid New Sources}
\label{sec:counting}

This section counts the valid new sources that are simultaneously at maximum graph distance from the origin and from a given node. Let $G_k$ be a dense Gaussian network of diameter $k$, and let
\[
A=(a_x,a_y)\in G_k.
\]
A node $P\in G_k$ is valid with respect to $0$ and $A$ when
\[
d(0,P)=k,
\qquad
 d(A,P)=k.
\]
Equivalently,
\[
P\in \partial G_k(0)\cap\partial G_k(A).
\]
We denote the number of such nodes by
\[
\operatorname{count}(A,k)=
\left|\{P\in G_k:d(0,P)=k,\ d(A,P)=k\}\right|.
\]

The important point is that the second boundary is taken in the quotient network. Hence, a boundary intersection may occur not only with the central translated copy $A+\partial G_k(0)$, but also with one of its neighboring quotient-lattice copies. The counting method below keeps the compact piecewise expression for the local intersection and then gives the exact quotient-lattice-aware formula used in validation.

\subsection{Boundary Partition}

Let
\[
D_k=\{(x,y)\in\mathbb Z^2: |x|+|y|\le k\}
\]
be the canonical diamond of representatives. Its boundary is
\[
\partial D_k=\{(x,y)\in\mathbb Z^2: |x|+|y|=k\}.
\]
We partition $\partial D_k$ into four disjoint directed sides:
\[
S_1=\{(t,k-t):0\le t\le k\},
\]
\[
S_2=\{(-t,k-t):1\le t\le k\},
\]
\[
S_3=\{(-t,-k+t):0\le t\le k-1\},
\]
\[
S_4=\{(t,-k+t):1\le t\le k-1\}.
\]
These sets are disjoint and have sizes $k+1$, $k$, $k$, and $k-1$, respectively. Therefore,
\[
|\partial D_k|=4k.
\]
This partition is used only as a counting device; distances are still interpreted in the Gaussian quotient network.

\begin{lemma}
\label{lem:boundary-intersection}
A node $P$ is a valid new source with respect to $0$ and $A$ if and only if
\[
P\in \partial G_k(0)\cap \partial G_k(A).
\]
Equivalently, after translating by $A$,
\[
P\in\partial G_k(0),
\qquad
P-A\in\partial G_k(0).
\]
\end{lemma}

\begin{proof}
By definition, validity means $d(0,P)=k$ and $d(A,P)=k$. The first condition is exactly $P\in\partial G_k(0)$. By translation invariance of graph distance in $G_k$,
\[
d(A,P)=d(0,P-A).
\]
Thus, $d(A,P)=k$ if and only if $P-A\in\partial G_k(0)$, or equivalently $P\in A+\partial G_k(0)=\partial G_k(A)$.
\end{proof}

\subsection{Local Piecewise Boundary Count}
\label{subsec:local-piecewise-count}

Before incorporating quotient-lattice copies, consider the local intersection obtained from the central copy only. Define
\[
\operatorname{count}_{0}(A,k)=
\left|\partial D_k\cap (A+\partial D_k)\right|.
\]
Let
\[
a=\max\{|a_x|,|a_y|\},
\qquad
b=\min\{|a_x|,|a_y|\}.
\]
Solving the sixteen side-pair intersections between the four sides of $\partial D_k$ and the four sides of $A+\partial D_k$ gives the following compact local count:
\[
\operatorname{count}_{0}(A,k)=
\begin{cases}
4k, & a=b=0,\\[3pt]
2k+2-a, & a=b,\\[3pt]
2k+1, & (a,b)=(1,0),\\[3pt]
k+4-a+I(a_xa_y<0),
&
\begin{array}{l}
a=b+1,\\
(a,b)\ne(1,0),
\end{array}
\\[6pt]
4, & a\ge b+2.
\end{cases}
\]
Here $I(\cdot)$ is $1$ when the condition is true and $0$ otherwise.

The five cases in this expression correspond to the five possible geometric overlap patterns between the two unshifted diamond boundaries. When $a=b=0$, the two centers coincide, so the two boundaries are identical and all $4k$ boundary nodes are counted. When $a=b>0$, the shift lies on a diagonal direction; in this case two long side-overlap intervals remain, and their combined length decreases linearly as the diagonal offset $a$ increases, giving $2k+2-a$. The adjacent-axis case $(a,b)=(1,0)$ is separated because the axis shift creates an additional endpoint contact that is not captured by the neighboring near-diagonal expression. When $a=b+1$, the two diamonds are one step away from diagonal alignment; the overlap consists of one main interval and several endpoint contacts, with one extra endpoint appearing exactly when $a_x$ and $a_y$ have opposite signs. Finally, when $a\ge b+2$, all interval overlaps disappear and only the four corner-to-side endpoint intersections remain.

\begin{table}[H]
\centering
\caption{Local unshifted boundary-intersection count by position type of $A$.}
\label{tab:side-intersection-count}
\begin{tabular}{c|c}
\hline
Condition on $A$ & Local count $\operatorname{count}_{0}(A,k)$ \\
\hline
$a=0,\ b=0$ & $4k$ \\
$a=b$ & $2k+2-a$ \\
$(a,b)=(1,0)$ & $2k+1$ \\
$a=b+1,\ (a,b)\ne(1,0)$ & $k+4-a+I(a_xa_y<0)$ \\
$a\ge b+2$ & $4$ \\
\hline
\end{tabular}
\end{table}

Table~\ref{tab:side-intersection-count} is useful because it explains the geometry of the boundary intersection. It also shows that, even before considering wrap-around copies, every nonzero node has at least four local candidates. However, the exact count in $G_k$ must include quotient-lattice copies, as described next.

\begin{figure}[H]
\centering
\includegraphics[width=0.6\linewidth]{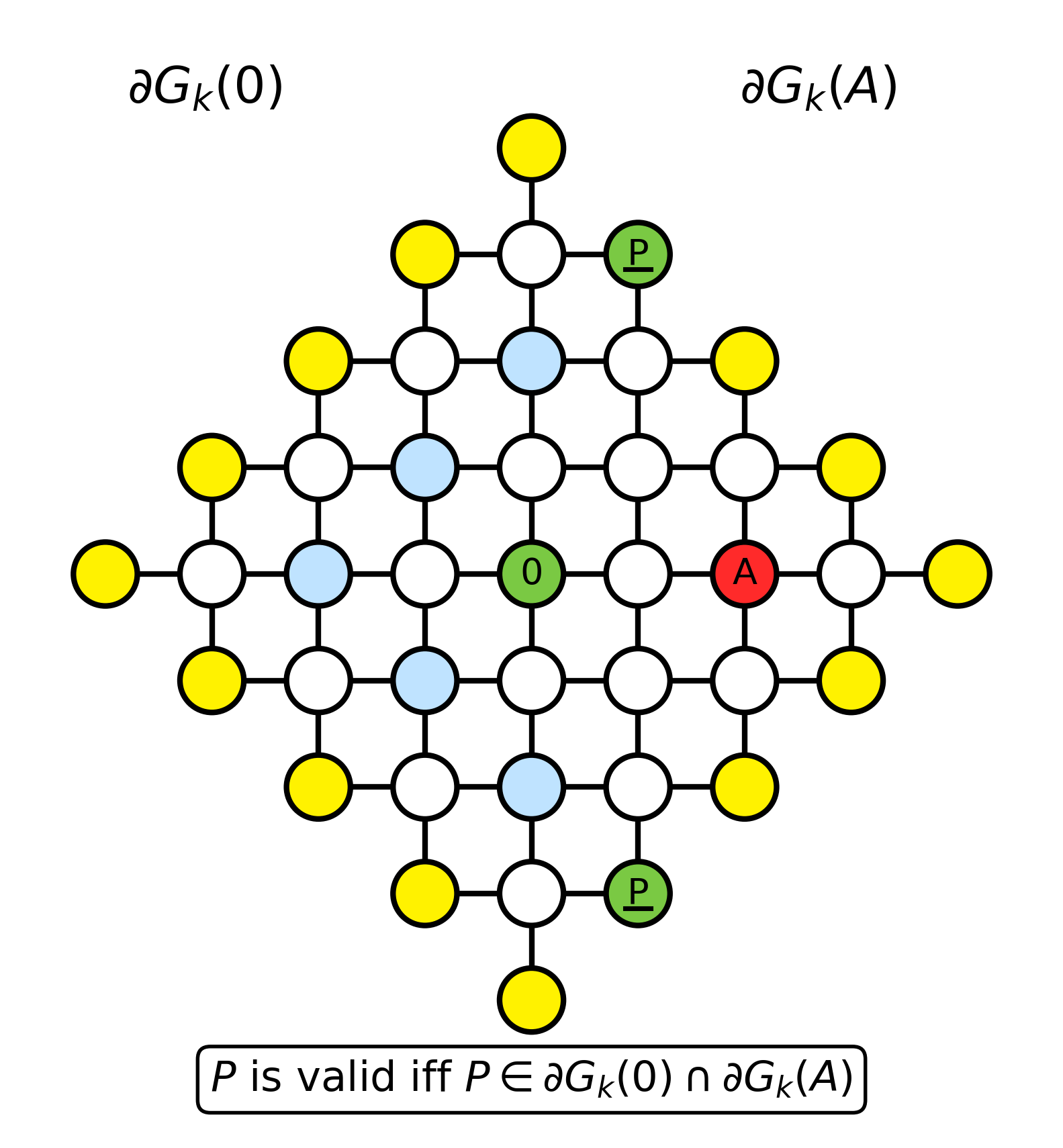}
\caption{Boundary-intersection view of valid new sources. Yellow nodes represent the boundary $\partial G_k(0)$, light-blue nodes represent the shifted boundary $\partial G_k(A)$, and green underlined nodes represent valid new sources in $\partial G_k(0)\cap\partial G_k(A)$.}
\label{fig:boundary_intersection}
\end{figure}

\subsection{Exact Quotient-Lattice-Aware Count}
\label{subsec:quotient-count}

The quotient lattice generated by $k+(k+1)i$ has basis vectors
\[
e_1=(k,k+1),
\qquad
e_2=(-(k+1),k).
\]
Only the central copy and the eight adjacent quotient-lattice copies can intersect the canonical boundary. Therefore, define the fixed shift set
\[
\mathcal L_k=
\{m e_1+n e_2: m,n\in\{-1,0,1\}\}.
\]
For each side $S_i$, let $P_i(t)$ be its linear parameterization and let $R_i$ be its integer parameter range. For a fixed shift $L\in\mathcal L_k$ and a side pair $(i,j)$, define
\[
T_{ij}(A,L)=
\{t\in R_i: P_i(t)=A+L+P_j(u)
\text{ for some }u\in R_j\}.
\]
Each set $T_{ij}(A,L)$ is either empty, a single integer, or an integer interval obtained by solving two linear equations in $t$ and $u$. Since multiple quotient copies can represent the same boundary node, duplicate parameter values on each side are removed by interval union.

\begin{theorem}
\label{thm:counting-formula}
For every node $A\in G_k$, the exact number of valid new sources in the Gaussian quotient network is
\[
\operatorname{count}(A,k)=
\sum_{i=1}^{4}
\left|
\bigcup_{L\in\mathcal L_k}
\bigcup_{j=1}^{4}
T_{ij}(A,L)
\right|.
\]
The formula uses a fixed number of side-pair interval computations and therefore gives an $O(1)$ counting procedure with respect to $k$ and $N$.
\end{theorem}

\begin{proof}
By Lemma~\ref{lem:boundary-intersection}, valid new sources are exactly the nodes in $\partial G_k(0)\cap\partial G_k(A)$. In the coordinate cover, $\partial G_k(A)$ is represented by the shifted copies $A+L+\partial D_k$, where $L$ ranges over quotient-lattice vectors. A copy can intersect the canonical boundary $\partial D_k$ only if it is one of the nine shifts in $\mathcal L_k$; all other copies are separated from $D_k$ by more than the boundary diameter in at least one lattice direction. Thus, every valid intersection is represented by some side pair $S_i$ and $A+L+S_j$ with $L\in\mathcal L_k$. For each fixed $(i,j,L)$, equating the two side parameterizations gives a pair of linear equations in $t$ and $u$, whose valid solutions form $T_{ij}(A,L)$. Taking the union over all $j$ and $L$ for each side $S_i$ removes duplicate representations, and summing over the four disjoint sides counts each boundary node exactly once.
\end{proof}

\begin{corollary}
\label{cor:min-four}
For every nonzero node $A\in G_k$,
\[
\operatorname{count}(A,k)\ge4.
\]
\end{corollary}

\begin{proof}
The exact quotient count contains the local central-copy contribution. From the local piecewise count in Table~\ref{tab:side-intersection-count}, every nonzero node has at least four local boundary intersections. Therefore, the exact quotient count is also at least four.
\end{proof}

\subsection{Examples for $k=5$}

For $k=5$, the boundary of the origin contains $4k=20$ nodes. The local count gives
\[
\operatorname{count}_{0}((0,0),5)=20,
\]
\[
\operatorname{count}_{0}((1,0),5)=2(5)+1=11,
\]
\[
\operatorname{count}_{0}((1,1),5)=2(5)+(2-1)=11,
\]
and
\[
\operatorname{count}_{0}((2,-1),5)=5+(4-2)+1=8.
\]
The quotient-lattice-aware formula in Theorem~\ref{thm:counting-formula} evaluates the exact count in $G_k$ by adding all valid shifted-copy intersections and removing duplicates.

\section{Direct New-Source Selection for Two Faulty Nodes}
\label{sec:direct-selection}

The previous section showed how valid new sources can be counted without scanning the network. This section gives a direct construction for selecting one such source for two arbitrary distinct faulty nodes.

Let
\[
A=(a_1,a_2),
\qquad
B=(b_1,b_2)
\]
be two distinct faulty nodes in $G_k$. We seek $NS\in G_k$ such that
\[
d(NS,A)=d(NS,B)=k.
\]
If $A=B$, the problem reduces to the one-fault case, where any node on $\partial G_k(A)$ is valid. Thus, the nontrivial case is $A\ne B$.

\subsection{Translation to the Origin}

Define the difference node
\[
C=\operatorname{mod}_{G_k}(B-A)=(c,d).
\]
Instead of searching directly for $NS$, we first find $P=(p,q)$ satisfying
\[
d(P,0)=k,
\qquad
 d(P,C)=k.
\]
Then the required new source is
\[
NS=\operatorname{mod}_{G_k}(A+P).
\]

\begin{figure}[H]
\centering
\includegraphics[width=0.7\linewidth]{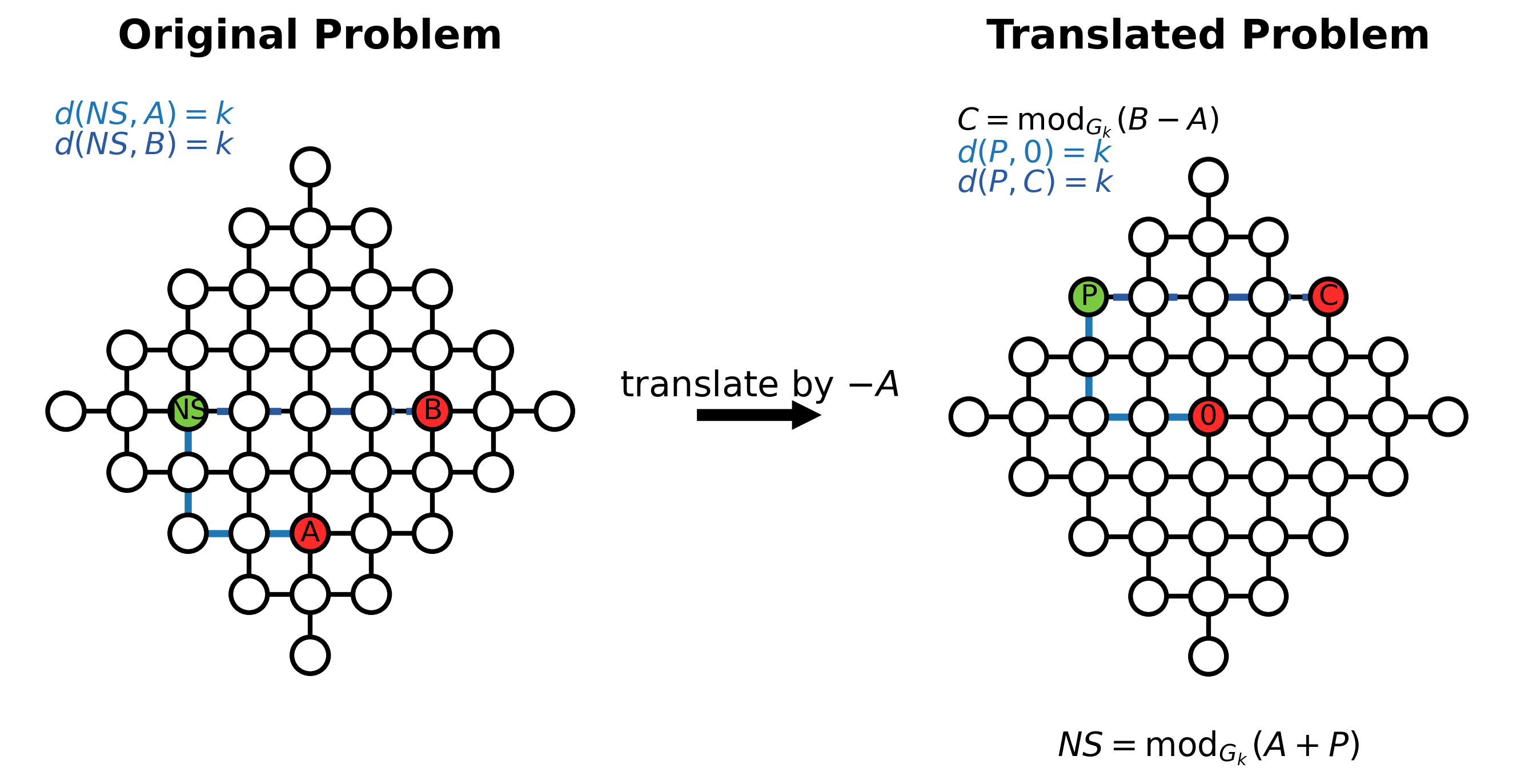}
\caption{Translation of the two-fault new-source problem. The original faulty-node pair $A,B$ is translated to $0,C$, where $C=\operatorname{mod}_{G_k}(B-A)$. After finding $P$ such that $d(P,0)=d(P,C)=k$, the source is shifted back as $NS=\operatorname{mod}_{G_k}(A+P)$.}
\label{fig:translation}
\end{figure}

\begin{lemma}
\label{lem:translation}
Let $A,B\in G_k$ and $C=\operatorname{mod}_{G_k}(B-A)$. If $P\in G_k$ satisfies $d(P,0)=d(P,C)=k$, then
\[
NS=\operatorname{mod}_{G_k}(A+P)
\]
satisfies $d(NS,A)=d(NS,B)=k$.
\end{lemma}

\begin{proof}
By translation invariance of graph distance,
\[
d(\operatorname{mod}_{G_k}(A+P),A)=d(P,0).
\]
Also, since $B=\operatorname{mod}_{G_k}(A+C)$,
\[
d(\operatorname{mod}_{G_k}(A+P),B)=d(P,C).
\]
The result follows from $d(P,0)=d(P,C)=k$.
\end{proof}

\subsection{Shifted Algebraic Construction of $P$}
\label{subsec:shifted-direct-selector}

A purely unshifted construction would solve
\[
|p|+|q|=k,
\qquad
|p-c|+|q-d|=k.
\]
This is not sufficient in the quotient network because the second boundary may intersect the first boundary through a neighboring quotient-lattice copy. Therefore, we use the fixed shift set
\[
\mathcal L_k=\{m(k,k+1)+n(-(k+1),k):m,n\in\{-1,0,1\}\}.
\]
For each $L=(L_x,L_y)\in\mathcal L_k$, set
\[
c'=c+L_x,
\qquad
d'=d+L_y.
\]
The selector solves
\[
|p|+|q|=k,
\qquad
|p-c'|+|q-d'|=k.
\]

The absolute values are removed by considering the signs of
\[
p,\quad q,
\qquad
p-c',\quad q-d'.
\]
For a sign tuple $(s_1,s_2,s_3,s_4)\in\{\pm1\}^4$, the signed system is
\[
s_1p+s_2q=k,
\]
\[
s_3p+s_4q=M,
\qquad
M=k+s_3c'+s_4d'.
\]
Let
\[
D=s_1s_4-s_2s_3.
\]
If $D\ne0$, Cramer's rule gives
\[
p=\frac{ks_4-s_2M}{D},
\qquad
q=\frac{s_1M-ks_3}{D}.
\]
The candidate is accepted only if $p$ and $q$ are integers and the two graph-distance conditions hold.

If $D=0$, the two signed equations are parallel. When $(s_3,s_4)=(s_1,s_2)$, consistency requires $M=k$ and the first boundary line is parameterized by
\[
p=s_1t,
\qquad
q=s_2(k-t),
\qquad
0\le t\le k.
\]
The second boundary imposes
\[
\max(0,s_1c')\le t\le \min(k,k-s_2d').
\]
When $(s_3,s_4)=(-s_1,-s_2)$, consistency requires $M=-k$ and the valid interval is
\[
\max(0,k-s_2d')\le t\le \min(k,s_1c').
\]
In either parallel case, a valid endpoint is selected if the interval is nonempty.

\begin{lemma}
\label{lem:nine-shifts}
For a canonical difference node $C\in D_k$, any intersection between $\partial D_k$ and a quotient copy of $C+\partial D_k$ is represented by one of the nine shifts in $\mathcal L_k$.
\end{lemma}

\begin{proof}
Let
\begin{align*}
e_1&=(k,k+1),\\
e_2&=(-(k+1),k),\\
L&=m e_1+n e_2 .
\end{align*}
Suppose that a shifted copy contributes an intersection. Then there exist $P,Q\in\partial D_k$ such that
\[
P=C+L+Q.
\]
Thus,
\[
C+L=P-Q.
\]
Since $P,Q\in D_k$, we have $\|P\|_1\le k$ and $\|Q\|_1\le k$, and hence
\[
\|C+L\|_1=\|P-Q\|_1\le \|P\|_1+\|Q\|_1\le 2k .
\]
Also, because $C$ is the canonical representative of a node in $G_k$, $C\in D_k$ and therefore $\|C\|_1\le k$. The reverse triangle inequality then gives the necessary condition
\[
\|L\|_1\le \|C+L\|_1+\|C\|_1\le 3k .
\]
Now compute the Manhattan length of a quotient-lattice shift:
\[
\|L\|_1
=
|mk-n(k+1)|+|m(k+1)+nk|.
\]
A direct case check on the signs of $m$ and $n$ shows that if $(m,n)$ is outside the block $\{-1,0,1\}^2$, then
\[
|mk-n(k+1)|+|m(k+1)+nk|>3k .
\]
Therefore, any quotient copy that can intersect the canonical boundary must have $m,n\in\{-1,0,1\}$, which gives exactly the central copy and the eight adjacent copies.
\end{proof}

\begin{algorithm}[H]
\caption{Shifted Constant-Time Direct New-Source Selection}
\label{alg:direct-selection}
\begin{algorithmic}[1]
\REQUIRE Network parameter $k$; faulty nodes $A,B\in G_k$
\ENSURE A valid new source $NS$
\STATE $C\gets \operatorname{mod}_{G_k}(B-A)$; write $C=(c,d)$
\STATE $e_1\gets(k,k+1)$, $e_2\gets(-(k+1),k)$
\FORALL{$m,n\in\{-1,0,1\}$}
\STATE $(L_x,L_y)\gets m e_1+n e_2$
\STATE $c'\gets c+L_x$, $d'\gets d+L_y$
\FORALL{$(s_1,s_2,s_3,s_4)\in\{\pm1\}^4$}
\STATE $M\gets k+s_3c'+s_4d'$
\STATE $D\gets s_1s_4-s_2s_3$
\IF{$D\ne0$}
\STATE $p\gets(ks_4-s_2M)/D$
\STATE $q\gets(s_1M-ks_3)/D$
\IF{$p,q$ are integers and $|p|+|q|=k$ and $|p-c'|+|q-d'|=k$}
\STATE $NS\gets\operatorname{mod}_{G_k}(A+(p,q))$
\IF{$d(NS,A)=k$ and $d(NS,B)=k$}
\STATE \textbf{return} $NS$
\ENDIF
\ENDIF
\ELSE
\STATE Apply the parallel interval rules for $c',d'$
\IF{a verified endpoint candidate $(p,q)$ is obtained}
\STATE \textbf{return} $\operatorname{mod}_{G_k}(A+(p,q))$
\ENDIF
\ENDIF
\ENDFOR
\ENDFOR
\end{algorithmic}
\end{algorithm}

\begin{theorem}[Correctness of the shifted direct-selection algorithm]
\label{thm:direct-selection}
For any two distinct faulty nodes $A,B\in G_k$, Algorithm~\ref{alg:direct-selection} returns a valid new source $NS$ satisfying
\[
d(NS,A)=d(NS,B)=k.
\]
\end{theorem}

\begin{proof}
The published re-rooting result guarantees that a common distance-$k$ node exists for every two-fault configuration in $G_k$. By Lemma~\ref{lem:nine-shifts}, the corresponding boundary intersection is represented by one of the nine shifts in $\mathcal L_k$. For that shift, the signs of $p,q,p-c',q-d'$ determine one of the sixteen signed systems. If the system is nonparallel, Cramer's rule recovers the candidate. If it is parallel, the interval rule recovers a candidate from the nonempty interval. The final verification accepts only candidates satisfying $d(P,0)=d(P,C)=k$. Lemma~\ref{lem:translation} then gives $d(NS,A)=d(NS,B)=k$.
\end{proof}

\subsection{Worked Example}
\label{subsec:direct-example}

Let $k=5$ and suppose the two faulty nodes are
\[
A=(1,1),
\qquad
B=(3,1).
\]
Then
\[
C=\operatorname{mod}_{G_5}(B-A)=(2,0).
\]
Using the unshifted copy $L=(0,0)$ and the sign tuple $(s_1,s_2,s_3,s_4)=(1,1,-1,1)$, we have
\[
M=5+(-1)(2)+(1)(0)=3,
\qquad
D=2.
\]
Thus,
\[
p=\frac{5(1)-1(3)}{2}=1,
\qquad
q=\frac{1(3)-5(-1)}{2}=4.
\]
So $P=(1,4)$, and
\[
|P|_1=|1|+|4|=5,
\qquad
|P-C|_1=|1-2|+|4|=5.
\]
Therefore $d(P,0)=d(P,C)=5$. Shifting back gives
\[
NS=\operatorname{mod}_{G_5}(A+P)=\operatorname{mod}_{G_5}(2,5)=(-3,-1).
\]
A direct distance check gives
\[
d(NS,A)=d(NS,B)=5,
\]
so $NS=(-3,-1)$ is a valid new source.

\subsection{Constant-Time Selection}
\label{subsec:constant-time-selection}

The shifted direct construction checks nine quotient-lattice shifts and sixteen sign tuples for each shift. Therefore, the maximum number of algebraic cases is
\[
9\times16=144.
\]
Each case uses a constant number of arithmetic operations. Under the standard word-RAM model, where arithmetic on node coordinates is treated as constant time, the selector runs in $O(1)$ time. This replaces the $O(k)$ boundary-search step used in the published re-rooting method with a constant-time algebraic selection step.

\section{Complexity Comparison}
\label{sec:complexity}

This section compares the proposed shifted direct selector with two search-based alternatives. Let $G_k$ contain
\[
N=k^2+(k+1)^2
\]
nodes and have diameter $k$.

A full node scan tests every node $X\in G_k$ and checks whether $d(X,A)=d(X,B)=k$. Its worst-case complexity is $O(N)$, equivalently $O(k^2)$. A boundary scan first translates the problem to the origin and then tests the $4k$ nodes of $\partial G_k(0)$, giving $O(k)$ complexity. The proposed method checks only nine quotient-lattice shifts and sixteen sign tuples for each shift, giving at most $144$ algebraic cases and $O(1)$ complexity.

\begin{table}[H]
\centering
\caption{Complexity comparison of new-source selection methods.}
\label{tab:complexity-comparison}
\begin{tabular}{c|c|c}
\hline
Method & Candidates or cases & Complexity \\
\hline
Full node scan & $N=k^2+(k+1)^2$ & $O(N)$ \\
Boundary scan & $4k$ & $O(k)$ \\
Proposed shifted selector & $9\times16=144$ & $O(1)$ \\
\hline
\end{tabular}
\end{table}

The comparison shows that the proposed selector has a fixed algebraic workload independent of both $k$ and $N$. This is the central algorithmic improvement over the published boundary-search source-selection step.

\section{Computational Validation}
\label{sec:validation}

This section validates the quotient-lattice-aware counting formula and the shifted direct new-source selector. The validation uses the dense Gaussian networks with
\[
k\in\{10,25,50,100,200\},
\qquad
N=k^2+(k+1)^2.
\]
All distance checks are performed using graph distance in $G_k$.

\subsection{Validation of the Counting Formula}

For each tested node $A\in G_k$, the quotient-lattice-aware count from Theorem~\ref{thm:counting-formula} was compared with brute-force counting over the $4k$ boundary nodes. For $k=10$, $25$, and $50$, all nodes were tested. For $k=100$ and $k=200$, 10,000 nodes were sampled. Table~\ref{tab:count-validation-final} shows zero mismatches and zero maximum error in every tested case.

\begin{table}[H]
\centering
\caption{Validation of the quotient-lattice-aware counting formula.}
\label{tab:count-validation-final}
\begin{tabular}{c|c|c|c|c}
\hline
$k$ & $N$ & Nodes tested & Mismatches & Max error \\
\hline
10  & 221   & 221   & 0 & 0 \\
25  & 1301  & 1301  & 0 & 0 \\
50  & 5101  & 5101  & 0 & 0 \\
100 & 20201 & 10000 & 0 & 0 \\
200 & 80401 & 10000 & 0 & 0 \\
\hline
\end{tabular}
\end{table}

\subsection{Validation of Direct New-Source Selection}

The second experiment validates the shifted direct selector for arbitrary distinct faulty nodes. For each value of $k$, 100,000 random fault pairs were tested. For every pair, the returned source was verified by checking
\[
d(NS,A)=k,
\qquad
 d(NS,B)=k.
\]
No boundary-search fallback was used. Table~\ref{tab:direct-validation-final} reports 500,000 valid outputs and zero failures. The maximum observed number of checked shifted sign cases was $51$, well below the theoretical bound of $144$.

\begin{table}[H]
\centering
\caption{Validation of shifted direct new-source selection.}
\label{tab:direct-validation-final}
\begin{tabular}{c|c|c|c|c|c}
\hline
$k$ & $N$ & Pairs & Valid & Failed & Max cases \\
\hline
10  & 221   & 100000 & 100000 & 0 & 51 \\
25  & 1301  & 100000 & 100000 & 0 & 51 \\
50  & 5101  & 100000 & 100000 & 0 & 51 \\
100 & 20201 & 100000 & 100000 & 0 & 51 \\
200 & 80401 & 100000 & 100000 & 0 & 51 \\
\hline
\end{tabular}
\end{table}

\subsection{Re-Rooting Broadcast Validation}

A full re-rooting validation was also performed using one- and two-node fault scenarios and four fault-placement modes: random, near-source, critical-position, and close-pair. For each $k$, each fault count, and each mode, 1000 trials were generated, giving 40,000 total raw trials. The proposed method succeeded in every trial and reached exactly $N-f$ non-faulty nodes, where $f$ is the number of faulty nodes. Table~\ref{tab:rerooting-summary-final} gives the results aggregated by $k$ and fault count.

\begin{table}[H]
\centering
\caption{Re-rooting broadcast validation aggregated across all four fault-placement modes.}
\label{tab:rerooting-summary-final}
\begin{adjustbox}{max width=\textwidth}
\begin{tabular}{c|c|c|c|c|c|c|c|c}
\hline
$k$ & $N$ & Faults & Trials & Base succ. (\%) & Prop. succ. (\%) & Avg. base reach & Avg. prop. reach & Avg. total steps \\
\hline
10  & 221   & 1 & 4000 & 5.275 & 100.000 & 202.452 & 220.000 & 18.418 \\
10  & 221   & 2 & 4000 & 2.350 & 100.000 & 194.500 & 219.000 & 17.815 \\
25  & 1301  & 1 & 4000 & 2.425 & 100.000 & 1209.836 & 1300.000 & 46.202 \\
25  & 1301  & 2 & 4000 & 0.725 & 100.000 & 1193.119 & 1299.000 & 44.491 \\
50  & 5101  & 1 & 4000 & 1.150 & 100.000 & 4855.743 & 5100.000 & 92.252 \\
50  & 5101  & 2 & 4000 & 0.325 & 100.000 & 4838.786 & 5099.000 & 88.730 \\
100 & 20201 & 1 & 4000 & 0.450 & 100.000 & 19302.131 & 20200.000 & 184.587 \\
100 & 20201 & 2 & 4000 & 0.100 & 100.000 & 19253.217 & 20199.000 & 177.260 \\
200 & 80401 & 1 & 4000 & 0.250 & 100.000 & 76962.268 & 80400.000 & 369.395 \\
200 & 80401 & 2 & 4000 & 0.075 & 100.000 & 76748.182 & 80399.000 & 354.469 \\
\hline
\end{tabular}
\end{adjustbox}
\end{table}

\subsection{Runtime Comparison}

The final experiment compares the original boundary-search selector with the proposed shifted direct selector. For each $k$, the runtime is averaged over 100,000 randomly sampled fault pairs. Table~\ref{tab:runtime-final} shows that the shifted selector has a small fixed overhead for small networks, but its runtime remains nearly stable as $k$ grows. It becomes faster for larger networks and reaches a $5.92\times$ speedup at $k=200$.

\begin{table}[H]
\centering
\caption{Average runtime per new-source selection query.}
\label{tab:runtime-final}
\begin{tabular}{c|c|c|c|c}
\hline
$k$ & $N$ & Boundary ms & Shifted ms & Speedup \\
\hline
10  & 221   & 0.003800 & 0.012471 & 0.30$\times$ \\
25  & 1301  & 0.011542 & 0.017004 & 0.68$\times$ \\
50  & 5101  & 0.024058 & 0.017400 & 1.38$\times$ \\
100 & 20201 & 0.043175 & 0.015310 & 2.82$\times$ \\
200 & 80401 & 0.113837 & 0.019244 & 5.92$\times$ \\
\hline
\end{tabular}
\end{table}

\begin{table}[H]
\centering
\caption{Search-space size or algebraic-case bound for each method.}
\label{tab:search-space-final}
\begin{tabular}{c|c|c|c}
\hline
$k$ & Full scan & Boundary scan & Shifted direct \\
\hline
10  & 221   & 40  & 144 \\
25  & 1301  & 100 & 144 \\
50  & 5101  & 200 & 144 \\
100 & 20201 & 400 & 144 \\
200 & 80401 & 800 & 144 \\
\hline
\end{tabular}
\end{table}

Overall, the validation confirms three points. First, the quotient-lattice-aware count matches brute force in all tested cases. Second, the shifted selector returned valid new sources for all 500,000 sampled pairs without fallback. Third, the complete re-rooting simulation achieved 100\% proposed recovery across 40,000 trials.

\section{Conclusion}
\label{sec:conclusion}

This paper extended the published re-rooting-based fault-tolerant broadcasting framework for dense Gaussian networks by replacing the $O(k)$ boundary-search source-selection step with a constant-time algebraic method. The counting problem was formulated as a boundary-intersection problem in the Gaussian quotient. A compact piecewise formula was retained for the local unshifted boundary intersection, and an exact quotient-lattice-aware counting formula was developed using a fixed union of side-pair intervals over nine quotient-lattice copies.

For two faulty nodes $A$ and $B$, the source-selection problem was translated to the difference node $C=\operatorname{mod}_{G_k}(B-A)$. The proposed shifted selector searches a fixed set of nine lattice shifts and sixteen sign configurations for each shift. Nonparallel systems are solved using Cramer's rule, while parallel systems are handled by constant-time interval-endpoint selection. Hence, at most $9\times16=144$ shifted sign cases are checked, and the method runs in $O(1)$ time under the standard word-RAM model.

Computational validation confirmed the theoretical claims. The quotient-lattice-aware count produced zero mismatches over 26,623 tested nodes. The shifted selector returned valid new sources for all 500,000 sampled fault pairs. In the re-rooting broadcast validation, the proposed method achieved 40,000 successful recoveries out of 40,000 trials and reached exactly $N-f$ non-faulty nodes in every configuration. Runtime results show that the shifted selector becomes faster than boundary search for larger networks, with a $5.92\times$ speedup at $k=200$.

Future work may investigate source ranking when multiple valid candidates exist, extend the method to link failures or dynamic faults, and study whether similar quotient-lattice-aware constant-time recovery rules can be developed for other algebraic interconnection networks.

\appendix
\section{Side-Pair Interval Counting Details}
\label{app:side-pair-counting}

This appendix gives the constant-size side-pair interval calculation used in Theorem~\ref{thm:counting-formula}. Each side $S_i$ is written as
\[
P_i(t)=o_i+v_i t,
\qquad
 t\in R_i,
\]
where
\[
\begin{array}{c|c|c|c}
 i & o_i & v_i & R_i \\
\hline
1 & (0,k) & (1,-1) & [0,k] \\
2 & (0,k) & (-1,-1) & [1,k] \\
3 & (0,-k) & (-1,1) & [0,k-1] \\
4 & (0,-k) & (1,1) & [1,k-1].
\end{array}
\]
For a node $A=(a_x,a_y)$, a quotient-lattice shift $L=(L_x,L_y)$, and a side pair $(i,j)$, an intersection point satisfies
\[
o_i+v_i t=A+L+o_j+v_j u.
\]
Equivalently,
\[
v_i t-v_j u=A+L+o_j-o_i.
\]
This is a $2\times2$ integer linear system in the side parameters $t$ and $u$. If the determinant is nonzero, there is at most one solution, and it is counted only when both parameters lie in their valid integer ranges. If the determinant is zero, the two side lines are parallel. The consistency condition is checked first; if the system is consistent, the valid values of $t$ form an integer interval after imposing the range restriction on $u$.

For each fixed first side $S_i$, the intervals obtained from all $j$ and all $L\in\mathcal L_k$ are merged before counting. This removes duplicate quotient representations of the same boundary node. The exact count is therefore
\[
\operatorname{count}(A,k)=
\sum_{i=1}^{4}
\left|
\bigcup_{L\in\mathcal L_k}
\bigcup_{j=1}^{4}
T_{ij}(A,L)
\right|,
\]
which is the formula used in Theorem~\ref{thm:counting-formula} and in the computational validation.

\section*{Acknowledgments}
The authors would like to acknowledge the support of Kuwait University and its Computer Science Department.

\end{document}